\documentclass[11pt]{amsart}

\usepackage{tikz}
\usepackage{csquotes}
\usepackage{amssymb}
\usepackage{stmaryrd}
\usepackage{subcaption}
\usepackage{algorithm}
\usepackage{algpseudocode}

\usepackage{fullpage}

\usepackage[T5,T1]{fontenc}
\usepackage[utf8]{inputenc}

\usepackage{hyperref}
\usepackage{cleveref}

\theoremstyle{plain}
\newtheorem{theorem}{Theorem}[section]
\newtheorem{lemma}[theorem]{Lemma}
\newtheorem{corollary}[theorem]{Corollary}
\newtheorem{proposition}[theorem]{Proposition}

\theoremstyle{definition}
\newtheorem{definition}[theorem]{Definition}

\theoremstyle{remark}
\newtheorem{remark}[theorem]{Remark}

\tikzstyle{vertex}=[circle,fill=black,minimum size=7pt,inner sep=0pt]
\tikzstyle{bigvertex}=[circle,draw,thick,,fill=black!20,minimum size=20pt,inner sep=0pt]
\tikzstyle{matching edge}=[blue, ultra thick]
\tikzstyle{non matching edge}=[red]
\definecolor{lavenderindigo}{rgb}{0.58, 0.34, 0.92}
\definecolor{amber}{rgb}{1.0, 0.75, 0.0}

\begin{document}

\title[Constrained path-finding and structure from acyclicity]{Constrained
  path-finding and structure from acyclicity}
\author{{\fontencoding{T5}\selectfont Lê Thành Dũng Nguyễn}}
\thanks{Part of this work was carried out when the author was a student at École
  normale supérieure de Paris. After the author joined the Laboratoire
  d'informatique de Paris Nord and continued working on this paper,
  he was at first partially supported by the now finished ANR
  project Elica (ANR-14-CE25-0005).}
\address{Laboratoire d'informatique de Paris Nord, Villetaneuse, France}
\email{nltd@nguyentito.eu}
\urladdr{https://nguyentito.eu/}

\maketitle

\begin{abstract}
  This note presents several results in graph theory inspired by the author's
  work in the proof theory of linear logic; these results are purely
  combinatorial and do not involve logic.
  
  We show that trails avoiding forbidden transitions, properly arc-colored
  directed trails and rainbow paths for complete multipartite color classes can
  be found in linear time, whereas finding rainbow paths is NP-complete for any
  other restriction on color classes. For the tractable cases, we also state new
  structural properties equivalent to Kotzig's theorem on the existence of
  bridges in unique perfect matchings.

  Another result on graphs equipped with unique perfect matchings that we prove
  here is the combinatorial counterpart of a theorem due to Bellin in linear
  logic: a connection between blossoms and bridge deletion orders.
\end{abstract}

\section{Introduction}

Many problems which consist of finding a \emph{path} or \emph{trail} (see next
page for the terminology) under some constraints between two given vertices are
equivalent to the \emph{augmenting path} problem for matchings, and thus
tractable. The many-one reductions involved often preserve structural
properties; Kotzig's theorem on the existence of bridges in unique perfect
matchings~\cite{kotzig_z_1959} thus yields equivalent \enquote{structure from
  acyclicity} theorems associated to these problems
(see~\cite{szeider_theorems_2004}). That is, the absence of constrained cycles
or closed trails entails the positive existence of some structure in the graph.
(Recall that a perfect matching is unique if and only if it admits no
\emph{alternating cycle}.)

A little-known fact is that this family of equivalent problems has a
representative in proof theory, a fact which has led researchers in logic to
independently rediscover equivalent versions of Kotzig's theorem. This
connection between the theory of \emph{proof net correctness} in \emph{linear
  logic}~\cite{girard_linear_1987} and perfect matchings was first noticed by
Retoré~\cite{retore_handsome_2003}, but remained unexploited for a long time. We
applied it in a previous paper~\cite{nguyen_unique_2018} in order to make
progress on open problems concerning proof nets. Conversely, here we demonstrate
some benefits of this bridge between logic and graphs for pure graph theory.
Some of the results we present here were directly motivated by logical
applications and are used either in~\cite{nguyen_unique_2018}, in its upcoming
journal version, or in our more recent work~\cite{nguyen_proof_2020}.

The paper can be divided into two distinct subthemes:
\begin{itemize}
\item The first concerns the structure of unique perfect matchings, in
  particular their \enquote{inductive} or \enquote{hereditary} characterization
  suggested by Kotzig's theorem. We relate this with the
  \emph{blossoms}~\cite{edmonds_paths_1965} which appear in combinatorial
  matching algorithms. This was discovered by directly transposing a theorem on
  proof nets, due to Bellin~\cite{bellin_subnets_1997}, using the reductions
  defined in our previous work~\cite{nguyen_unique_2018}; it turned out that a
  notion of \enquote{dependency} investigated by logicians corresponded exactly
  to blossoms. The result already appeared in the conference version of the
  latter paper, but we present it here in a stand-alone way for a broader
  audience of graph theorists, using the direct proof
  from~\cite[Appendix~D]{nguyen_unique_2018}.
\item The second is a study of some constrained path-finding problems, which
  covers the results that we have previously announced at the
  CTW'18\footnote{Cologne-Twente Workshop on Graphs and Combinatorial
    Optimization 2018.} workshop. We show, for most of these studied problems,
  that they belong to the aforementioned family. This leads to algorithms to
  solve them and \enquote{structure from acyclicity} properties. In one case, we
  establish a \enquote{dichotomy theorem} between a tractable case -- which is
  indeed equivalent to augmenting paths in matchings -- and a NP-complete case.
  A key ingredient in all of those results is a simple \emph{edge-colored line
    graph} construction; here the input of proof nets is that this construction
  was discovered by dissecting Retoré's aforementioned
  work~\cite{retore_handsome_2003}, in which it implicitly occurs. We also
  prove, at the end of the paper, a simple NP-completeness result that does not
  rely on this construction.
\end{itemize}

The first item will be treated in \cref{sec:blossoms}, while the second one --
whose results -- will be
the subject of the rest of the paper. For the remainder of the introduction, we
will explain in more detail the contents of this second item, with the following
table providing a summary of the discussion. (Our contributions, marked in bold,
fill some gaps and thus answer several natural questions.)

  \begin{center}
      \begin{tabular}{ll}
        \hline
        & Time complexity / additional results \\
        \hline
        Path avoiding forbidden transitions & NP-complete with dichotomy result~\cite{szeider_finding_2003} \\
        Trail avoiding forbidden transitions & \textbf{Linear with structural theorem} \\
        Properly colored path & Linear with structural theorem (cf.\ \cite{bang-jensen_digraphs._2009}) \\
        Properly colored trail & Linear with structural theorem~\cite{abouelaoualim_paths_2008} \\
        Rainbow path/trail\footnotemark{} & NP-complete~\cite{chakraborty_hardness_2011}, \textbf{with dichotomy result}\\
        & \textbf{and structural theorem for the tractable case}\\
        Properly colored directed path & NP-complete~\cite{gourves_complexity_2013} \\
        Properly colored directed trail & \textbf{Linear} (already polynomial in~\cite{gourves_complexity_2013}) \\
        Alt.\ directed path for matching\footnotemark{} & \textbf{NP-complete}\\
        \hline
      \end{tabular}
  \end{center}

\addtocounter{footnote}{-1}
\footnotetext{The existence of a rainbow path is equivalent to the existence of a rainbow trail between two vertices.}
\addtocounter{footnote}{1}
\footnotetext{Alternating directed trails for a perfect matching -- with the
  right definition of the latter in a directed graph -- are a special case of
  properly colored directed trail, so the linear algorithm mentioned in the
  penultimate row applies. However, NP-hardness works the other way around
  (while membership in NP is always trivial here): our result on alternating
  directed paths is slightly stronger than the one on properly colored paths
  in~\cite{gourves_complexity_2013} (though we derive the former quite easily
  from the latter).}

\textbf{An important terminological note:} following a common usage (see e.g.\
\cite[Section~1.4]{bang-jensen_digraphs._2009}), a \emph{path} is a walk without
repeating \emph{vertices} and a \emph{trail} is a walk without repeating
\emph{edges}; a \emph{cycle} (resp.\ \emph{closed trail}) is a closed walk
without repeating vertices (resp.\ edges). Paths (resp.\ cycles) are trails
(resp.\ closed trails), but the converse does not always hold.

\subsection{Edge-colored graphs}

From an assignment of colors to the edges of a graph, one can define either
\emph{local} or \emph{global} constraints:

\begin{itemize}
\item In a \emph{properly colored} (PC) path (see \cite[Chapter
  16]{bang-jensen_digraphs._2009}) or trail (see
  \cite{abouelaoualim_paths_2008}), \emph{consecutive} edges must have different
  colors. Both can be found in linear time by reduction to augmenting paths, and
  conversely augmenting paths are a special case of both these problems. The
  structural result for PC cycles is Yeo's theorem on cut vertices separating
  colors~\cite[\S 16.3]{bang-jensen_digraphs._2009}.
\item In a \emph{rainbow} (also called \emph{heterochromatic} or
  \emph{multicolored}) path, \emph{all} edges have different colors. Finding a
  rainbow path is NP-complete~\cite{chakraborty_hardness_2011} in the general
  case; there is also an algorithm running in $2^k n^{O(1)}$ time and polynomial
  space for $k$ colors and $n$ vertices~\cite{kowalik_finding_2016}. Let us also
  mention the related subject of \emph{rainbow connectivity}\footnote{An
    edge-colored graph is rainbow connected if all pairs of vertices can be
    joined by a rainbow path; a common question is coloring the edges of a graph
    with a minimum number of colors to make it rainbow connected.}, that has
  been an active area of research (see e.g.\ the PhD
  thesis~\cite{lauri_chasing_2016}) since its introduction in
  \cite{chartrand_rainbow_2008}.
\end{itemize}

For rainbow paths, we investigate whether restrictions on the shape of the
\emph{color classes} -- that is, the subgraphs induced by all edges of a given
color -- make the problem tractable, and we establish a dichotomy: there is a
single case which is not NP-hard (\cref{dichotomy-thm}), and it can be solved in linear
time. This tractable case also exhibits structure from acyclicity
(\cref{rainbow-structure}).

A special case of this structural theorem had been previously proved in Retoré's
PhD thesis \cite[Chapter~2]{retore_reseaux_1993}, for a class of graphs he
called \enquote{aggregates}, in an early attempt to extract the graph-theoretic
content of the theory of proof net correctness. Aggregates are edge-colored
graphs whose color classes are complete bipartite; we will call them
\emph{bipartite decompositions} instead,
following~\cite{alon_multicolored_1991}. Our result generalizes this to
\emph{multipartite decompositions} (\cref{def:multidec}).

\subsection{Forbidden transitions}

A very general notion of \emph{local} constraints is to simply forbid some pairs
of edges from occuring consecutively in a path -- what has been called a
\emph{path avoiding forbidden transitions} (we will also speak of
\emph{compatible paths}). Finding a compatible \emph{path} has been proven to be
NP-complete in general~\cite{szeider_finding_2003}, with a dichotomy theorem
(the tractable case covers in particular properly colored cycles). However, the
question for compatible \emph{trails} does not seem to have been asked before in
its full generality (despite previous work on \emph{Eulerian trails} avoiding
forbidden transitions~\cite[Chapter~VI]{fleischner_eulerian_1990}).

We show that it is, again, part of our family of equivalent problems; more
precisely compatible trails can be found with time complexity \emph{linear} in
the number of \emph{allowed} transitions (\cref{compatible-trail}). The
associated structural result (\cref{forbidden-transition-bridge}) entails, as a
corollary, a new proof of the one for properly colored
trails~\cite[Theorem~2.4]{abouelaoualim_paths_2008}.



\subsection{Arc-colored directed graphs}

The notion of properly colored path/trail also makes sense in \emph{directed
  graphs} equipped with a coloring of their arcs (directed edges). We first look
at PC directed \emph{trails} and show that they can be found by a simple
breadth-first search algorithm, in linear time (\cref{thm:lolbfs}). There was
already a polynomial time algorithm for this problem in the
literature~\cite[Theorem~1]{gourves_complexity_2013}; we examine more precisely
its asymptotic complexity and argue that it is worse than ours
(\cref{rem:reload}).

In contrast, deciding the existence of a PC directed path is NP-complete, even
when we assume that the input has no PC circuit, as shown
in~\cite[Theorem~5]{gourves_complexity_2013}. From this, we deduce that finding
an alternating \emph{circuit} (i.e.\ directed cycle) for a (certain notion of)
perfect matching in a directed graph is NP-complete (\cref{thm:alt-circuit-np})
(and as a corollary, this is also the case for directed paths). Though our
reduction is pretty straightforward, and the statement itself might seem
unsurprising, we have proved it as a step in a possibly much more significant
result that we obtain in~\cite{nguyen_proof_2020}: the refutation of an old
conjecture in proof theory, namely the equivalence between pomset logic and
system BV, see~\cite{guglielmi_system_2007}\footnote{Though this
  paper~\cite{guglielmi_system_2007} was published in 2007, it has its origins
  in a technical report from the late 1990's which already contained this
  conjecture.} (though at the time of writing, we are not entirely sure that our
refutation is correct). This is why we propose in~\cite{nguyen_proof_2020} an
alternative and more direct NP-hardness proof for alternating circuits by
reduction from CNF-SAT, in order to make the proofreading of this latter paper
more self-contained and avoid relying on~\cite{gourves_complexity_2013}.

\subsection{The edge-colored line graph construction}

The construction underlying these results (except for those on arc-colored
digraphs) is the following. Given a set of forbidden transitions on a graph $G$,
one can consider the subgraph of its line graph $L(G)$ containing only the edges
corresponding to allowed transitions. To keep the information of the vertices in
$G$, one adds a coloring of the edges of this subgraph of $L(G)$: this is what
we call the \emph{edge-colored line graph} $L_{EC}(G)$ (\cref{def:lec}).

The results on trails avoiding forbidden transitions immediately follow from the
properties of $L_{EC}(G)$ together with the known results on PC paths. For this
specific purpose one can also use a variant $L_{PM}(G)$ defined using perfect
matchings (\cref{def:lpm}). As will be explained in the journal version
of~\cite{nguyen_unique_2018}, this $L_{PM}$ construction is at work implicitly
Retoré's paper~\cite{retore_handsome_2003}, and indeed, it is by attempting to
understand and generalize Retoré's reduction from proof nets to perfect
matchings that we were led to define the edge-colored line graph -- although it
could undoubtedly have been discovered without this inspiration, since it seems
to be a rather natural construction.

As for the dichotomy theorem for rainbow paths, it also relies mainly on the
$L_{EC}$ construction, combined with proof techniques
from~\cite{szeider_finding_2003} and~\cite{chakraborty_hardness_2011}, in
particular a characterization of complete multipartite graphs by excluded
vertex-induced subgraphs~\cite[Lemma~7]{szeider_finding_2003}.

\subsection{Other contributions}

For some of the aforementioned constraints, we investigate the problem of
finding a constrained path/trail visiting some prescribed intermediate
vertex/edge. The general pattern is NP-completeness in the general case and
tractability under an acyclicity assumption. The initial motivation for this was
as a subroutine to compute the \enquote{kingdom ordering} on unique perfect
matchings studied in \cref{sec:blossoms}, which has a rather meaningful logical
counterpart (see~\cite{nguyen_unique_2018}) as already mentioned.

We also make minor contributions to the theory of properly colored paths and
cycles, using slight variations on pre-existing proofs, in \cref{sec:remarks}.
Specifically, we show how to find properly colored cycles in linear time, and
generalize a reduction from 2-edge-colored graphs to matchings.

\subsection{Acknowledgments}

Thanks a lot to Christoph Dürr for drawing my attention to the topic of
constrained path-finding (accidentally, through his enthusiasm for programming
contest problems), and for his feedback on early iterations (2016--2017) of this
work. Thanks also to {\fontencoding{T5}\selectfont Nguyễn Kim Thắng} for letting
me work on this side project during my internship with Christoph and him.

\tableofcontents

\section{Bridges and blossoms in unique perfect matchings}
\label{sec:blossoms}

As mentioned in the introduction, we present here a new graph-theoretic result
directly coming from the theory of proof nets. We will also take the
opportunity, in \cref{subsec:decomposing} and \cref{subsec:algo-aspects}, to
recall well-known properties of perfect matchings which the rest of the paper
will rely on.

Let $G = (V,E)$ be a graph. Recall that a \emph{matching} $M$ of $G$ is a subset
of edges $M \subseteq E$ such that each vertex is incident to at least one edge.
If each vertex is incident to \emph{exactly one} edge, then $M$ is a
\emph{perfect matching}. We now recall a few properties of \emph{unique} perfect
matchings.

\begin{lemma}[Berge's lemma for cycles\footnotemark{}] 
  \label{berge-upm}
  Let $G$ be a graph and $M$ be a perfect matching of $G$. If $C$ is an
  \emph{alternating cycle} for $M$ -- i.e.\ $C$ alternates between edges in $M$
  and outside of $M$ -- then the symmetric difference $M \triangle C$ is another
  perfect matching.

  Conversely, if $M' \neq M$ is a different perfect matching, then $M \triangle
  M'$ is a vertex-disjoint union of cycles, which are alternating for both $M$
  and $M'$.

  Thus, a perfect matching is unique if and only if it admits no alternating
  cycle.
\end{lemma}
\footnotetext{Though this statement does not appear in~\cite{berge_two_1957}, it
  is a simple variant of Berge's lemma for paths (\cref{berge-augmenting}).}

\begin{theorem}[Kotzig~\cite{kotzig_z_1959}]
  \label{kotzig}
  Let $G$ be a graph with a \emph{unique} perfect matching $M$. Then $M$
  contains a \emph{bridge} of $G$, i.e.\ an edge whose removal incrases the
  number of connected components.
\end{theorem}

\subsection{Decomposing unique perfect matchings}
\label{subsec:decomposing}

Kotzig's theorem gives us a simple way to check that $M$ is the only perfect
matching in $G$: find some $(u,v) \in M$ which is a bridge of $G$, delete its
two endpoints from the graph, and repeat on the vertex-induced subgraph $G[V
\setminus \{u,v\}]$, in which $M \setminus \{(u,v)\}$ is a perfect matching. The
matching is unique if and only if we end up with the empty graph. Equivalently,
we may state:

\begin{proposition}
  Among the set of pairs $(G,M)$ of a graph $G$ and a perfect matching $M$ of
  $G$, the subset of those for which $M$ is the unique perfect matching of $G$
  is the smallest subset:
  \begin{itemize}
  \item containing the empty graph;
  \item such that, if $(u,v) \in M$ is a bridge in $G = (V,E)$ and $(G[V
    \setminus \{u,v\}], M \setminus \{(u,v)\})$ is in the subset, then $(G,M)$
    also is.
  \end{itemize}
\end{proposition}

\begin{remark}
  \enquote{The smallest subset such that\ldots} suggests an inductive,
  \enquote{bottom-up} characterization: the graphs equipped with unique perfect
  matchings are obtained from the empty graph by successively adding matching
  edges and potentially joining their endpoints to different connected
  components. Precise statements of this kind are given
  in~\cite[Theorem~1]{retore_handsome_2003},
  \cite[Theorem~2.4]{nguyen_unique_2018}.
\end{remark}

\begin{remark}
  This iterative bridge deletion procedure admits a variant which \emph{does not
    look at $M$}: for $G = (V,E)$, choose a bridge $(u,v) \in E$ such that the
  connected components of $u$ and $v$ in $(V, E \setminus \{(u,v)\})$ have an
  odd number of vertices (equivalently, the two new connected components in $G[V
  \setminus \{u,v\}]$ disconnected by the removal of $u$ and $v$ have an even
  number of vertices). If, by doing so iteratively, one reaches the empty graph,
  then the set of deleted bridges is the unique perfect matching of $M$. This
  gives a quasi-linear time algorithm for \emph{finding a unique perfect
    matching}, see \cite[\S2]{gabow_unique_2001} and
  \cite[\S1.2]{holm_dynamic_2018}.
\end{remark}

We will be interested in the order of deletion of the matching edges for some
execution of this procedure reaching the empty graph.

\begin{definition}
  Let $G = (V,E)$ be a graph with a unique perfect matching $M$. A \emph{bridge
    deletion ordering} is an ordering of the matching edges $M = \{(u_1, v_1),
  \ldots, (u_n, v_n)\}$ ($|V| = 2n$) such that for all $i \in \{1, \ldots, n\}$,
  $(u_i, v_i)$ is a bridge in $G[V \setminus \{u_1, v_1, \ldots, u_{i-1},
  v_{i-1}\}]$.

  We define the \emph{kingdom ordering} $\prec$, which is a partial order, as
  follows: $e \prec f$ if $e$ occurs before $f$ in all bridge deletion
  orderings of $M$ ($e, f \in M$).
\end{definition}

Bridge deletion orderings in unique perfect matchings are somewhat analogous to
perfect elimination orders in chordal graphs, with bridges instead of simplicial
vertices. The terminology \enquote{kingdom ordering} is lifted straight from
linear logic (see~\cite{bellin_subnets_1997}); it comes from the notion of
\emph{kingdom} whose graph-theoretic version we now define.

\begin{definition}
  Let $e \in M$. The result of iteratively deleting the endpoints of bridges in $M$
  \emph{except for $e$} is a non-empty vertex-induced subgraph of $G$ called the
  \emph{kingdom} of $e$.
\end{definition}

The idea is that the endpoints of $e$ must be deleted before any other matching
edge in the kingdom can be deleted in the decomposition by bridge deletion.
Thus, $e \prec f$ if and only if $f$ is in the kingdom of $e$, or equivalently,
iff the kingdom of $f$ is included in the kingdom of $e$.

\subsection{Blossoms and Bellin's theorem}

A \emph{blossom}~\cite{edmonds_paths_1965} is a cycle such that all its vertices
are matched within the cycle except for one, its \emph{root}. The \emph{stem} of
a blossom is the matching edge incident to its root. We will characterize the
kingdom ordering using the following notion:

\begin{definition}
  We say that $e$ \emph{blossom-binds} $f$, and write $e \rightarrow f$, when
  $e$ is the stem of some blossom $C$ such that $f \in M \cap C$.

  We write $\rightarrow^+$ (resp.\ $\rightarrow^*$) for the transitive (resp.\
  reflexive transitive) closure of $\rightarrow$.
\end{definition}

The main theorem of this section can now be stated:

\begin{theorem}
  \label{bellin-thm-graph}
  Let $M$ be a unique perfect matching. $\forall e,f \in M,\; e \prec f
  \Longleftrightarrow e \rightarrow^+ f$.
\end{theorem}

In a previous paper~\cite{nguyen_unique_2018}, we explained why this was
equivalent to Bellin's theorem on proof nets~\cite{bellin_subnets_1997}. This
equivalence already establishes the truth of the above theorem, but we will give
a direct proof without reference to logic. Indeed it was noted
in~\cite{nguyen_unique_2018} that translating to the setting of perfect matching
simplifies the statement -- which now involves the transitive closure of a
single relation, instead of two -- and so, accordingly, our proof should be
simpler than Bellin's.

\begin{proof}
  If $e \rightarrow f$ then $e \prec f$: $f$ cannot become a bridge as long as
  the cycle $C$ survives, and the only way to cut $C$ from the outside is to
  delete~$e$. By induction this establishes ($\Longleftarrow$).

  Conversely, let $G$ be a graph with a unique perfect matching $M$. We may
  assume w.l.o.g.\ that $e$ is the only bridge of $G$ in $M$, by restricting to
  the kingdom of $e$. Then $e$ is minimum for $\prec$ and the goal becomes to
  show that for all $f \neq e$, $e \rightarrow^+ f$. Removing the edge $e$,
  \emph{but not its endpoints}, results in two connected components which both
  have a unique \emph{near-perfect matching} (leaving one vertex unmatched)
  containing no bridge. If both these components have a single vertex, then the
  theorem is vacuously true; else, we have reduced it to the following
  \cref{near-perfect-bellin}.
\end{proof}

\begin{definition}
  We say that a vertex $u$ blossoms-binds a matching edge $f$, which we write $u
  \twoheadrightarrow f$, when $f$ is contained in a blossom with root $u$.
\end{definition}


\begin{proposition}
  \label{near-perfect-bellin}
  Let $G$ be a graph with a near-perfect matching $M$ and let $u$ be the
  unmatched vertex. Suppose $G$ has no bridge in $M$ and no alternating cycle
  for $M$. Then for all $f \in M$, there exists $g \in M$ such that $u
  \twoheadrightarrow g \rightarrow^* f$.
\end{proposition}

The proof of this proposition relies on the \emph{blossom shrinking} operation:
starting from the graph $G$ with a matching $M$, this consists in taking the
quotient graph $G'$ where all the vertices of the blossom have been identified;
$M$ induces a matching $M'$ in $G'$. This operation is also central in
combinatorial matching algorithms, starting with Edmonds's blossom
algorithm~\cite{edmonds_paths_1965}.

\begin{lemma}
  \label{lemma-bellin}
  Under the hypotheses of the proposition, if $M \neq \emptyset$, then:
  \begin{enumerate}
  \item There exists a blossom in $G$ for $M$ with root $u$.
  \item Let $G'$ be the graph obtained by shrinking this blossom, with induced
    matching $M' \subset M$. There is no bridge in $M'$ and no alternating cycle
    for $M'$.
  \item Let $u'$ be the exposed vertex in $G'$, corresponding to the shrunk
    blossom. For all $f \in M'$ with $u' \twoheadrightarrow f$ in $G'$, there
    exists $g \in M$ such that $u \twoheadrightarrow g \rightarrow^* f$ in $G$.
  \end{enumerate}
\end{lemma}

\begin{proof}[Proof of \cref{near-perfect-bellin} using \cref{lemma-bellin}]
  By induction on the size of $G$.

  Let us take a blossom using lemma (1). If it contains $f$, then $u
  \twoheadrightarrow f$ and we are done. Else, we shrink the blossom and get $G'$,
  $M'$ and $u'$; by lemma (2), they satisfy the assumptions of the proposition.
  By the induction hypothesis, there exists $g$ such that $u' \twoheadrightarrow g
  \rightarrow^* f$ in $G'$. Thanks to lemma (3), $u' \twoheadrightarrow g$ entails
  $u \twoheadrightarrow h \rightarrow^* g$ in $G$ for some $h \in M$. Also, $g
  \rightarrow^* f$ in $G'$ entails $g \rightarrow^* f$ in $G$ because the
  (possibly empty) sequence of blossoms which binds $f$ to $g$ in $G'$ cannot
  contain the vertex $u'$, and therefore lifts to exactly the same edges in $G$.
  Thus, $u \twoheadrightarrow h \rightarrow^* g \rightarrow^* f$ and
  therefore $u \twoheadrightarrow h \rightarrow^* f$ in $G$.
\end{proof}

\begin{proof}[Proof of \cref{lemma-bellin} (1)]
  The absence of alternating cycle amounts to saying that $M$ is the unique
  perfect matching of $G[V \setminus \{u\}]$ where $V$ is the vertex set of $G$.
  (Note that $M \neq \emptyset \Leftrightarrow V \setminus \{u\} \neq
  \emptyset$.) By Kotzig's theorem, $M$ contains a bridge $e$ of $G[V \setminus
  \{u\}]$; let $V_1$ and $V_2$ be the connected components created by the
  removal of $e$ (but keeping its endpoints) from $G[V \setminus \{u\}]$. We
  create a new graph $H$ by starting from $G[V \setminus \{u\}]$, adding two new
  vertices $u_1$ and $u_2$, and adding the edges $(u_i, v)$ for all $v \in V_i$
  with $v$ adjacent to $u$ in $G$ ($i = 1,2$), and the edge $(u_1, u_2)$.

  The perfect matching $M \cup \{(u_1, u_2)\}$ of $H$ contains no bridge of $H$:
  since $e$ is not a bridge of $G$, there is at least one edge between $u_1$ and
  $V_1$ and one edge between $u_2$ and $V_2$, so $(u_1, u_2)$ is not a bridge in
  $H$; and any edge of $M$ would be a bridge of $G$ if it were a bridge of $H$.
  Let us apply Kotzig's theorem again: this perfect matching admits an
  alternating cycle, which cannot be contained in $H[V \setminus \{u\}] = G[V
  \setminus \{u\}]$. Therefore, it contains an alternating path from $u_1$ to
  $u_2$, from which we retrieve a blossom with root $u$ in $G$.
\end{proof}

\begin{proof}[Proof of \cref{lemma-bellin} (2)]
  If there existed a bridge $e \in M'$ of $G'$, then $G' \setminus \{e\}$ would
  be disconnected while $G \setminus \{e\}$ would be connected; this is
  impossible. An alternating cycle for $M'$ would not visit $u'$ because it is
  unmatched, and therefore would be an alternating cycle for $M$ in $G$.
\end{proof}

\begin{proof}[Proof of \cref{lemma-bellin} (3)]
  Let $B$ be the blossom with root $u$ in $G$ that has been shrunk, and $B'$ be
  the blossom with root $u'$ in $G'$ containing $f$. There are two non-matching
  edges $e'_1$ and $e'_2$ in $B'$ incident to $u'$; let $e_1 = (u_1, v_1)$ be a
  preimage of $e'_1$ and $e_2 = (u_2, v_2)$ be a preimage of $e'_2$ in $G$, with
  $u_1, u_2 \in B$.

  The blossom $B$ can be decomposed into $P_1 \cup Q \cup P_2$, where $P_1$ is
  an alternating path from $u$ to $u_1$ (possibly empty, if $u = u_1$), $Q$ is
  an alternating path from $u_1$ to $u_2$ (possibly empty, if $u_1 = u_2$), and
  $P_2$ is an alternating path from $u_2$ to $u$. As for $B'$, it lifts to an
  alternating path $R$ between $u_1$ and $u_2$ starting and ending with a
  non-matching edge, so that $|R|$ is odd and $f \in R$. We proceed by case
  analysis on the parity of $|P_1|$ and $|P_2|$.
  \begin{itemize}
  \item If they are both even, then $P_1 \cup R \cup P_2$ is a blossom: $u
    \twoheadrightarrow f \rightarrow^* f$.
  \item If $|P_1|$ is even and $|P_2|$ is odd, then $Q \cup R$ is a blossom with
    root $u_1$. Either $u_1 = u$ and then $u \twoheadrightarrow f$, or there is an
    edge $g \in B \cap M$ incident to $u_1$ and then $u \twoheadrightarrow g
    \rightarrow f$.
  \item The case $|P_1|$ even and $|P_2|$ odd is symmetric to the previous one.
  \item If they were both odd, $Q \cup R$ would be an alternating cycle. \qedhere
  \end{itemize}
\end{proof}

\subsection{Algorithmic aspects}
\label{subsec:algo-aspects}

One interest of characterizing the kingdom ordering through blossom-binding is
that the later can be decided in polynomial time; thus, it shows that computing
the kingdom ordering is tractable.

To expound on this, we need to recall some bits of the classical theory around
matching algorithms. We start with a version of Berge's lemma which is relevant
to the \emph{maximum matching} problem, complementing~\cref{berge-upm}.

\begin{lemma}[Berge's lemma for paths~{\cite{berge_two_1957}}]
  \label{berge-augmenting}
  Let $G$ be a graph and $M$ be a matching of $G$. If $P$ is an \emph{augmenting
    path} for $M$ -- i.e. an alternating path whose endpoints are unmatched --
  then $M \triangle P$ is a matching and $|M \triangle P| = |M| + 1$.
  Conversely, if $M$ is a matching with $|M'| > |M|$, then $M \triangle M'$ is a
  vertex-disjoint union of:
  \begin{itemize}
  \item $|M'| - |M|$ augmenting paths for $M$;
  \item some (possibly zero) cycles which are alternating for both $M$ and $M'$.
  \end{itemize}
\end{lemma}

\begin{theorem}
  \label{augmenting-path-linear}
  There are algorithms with running time linear in the edges of the input for:
  \begin{itemize}
  \item finding an augmenting path for a matching, or detecting that the
    matching is maximum;
  \item deciding whether a perfect matching is unique, and finding an
    alternating cycle if it is not.
  \end{itemize}
\end{theorem}
\begin{proof}
  See \cite{gabow_linear-time_1985} and {\cite[Section 9.4]{tarjan_data_1983}}
  for the first result, \cite{gabow_unique_2001} for the second. Note that both
  of these algorithms use blossom shrinking.
\end{proof}

Since a significant portion of the algorithms we will see in this paper
ultimately proceed by a reduction to an augmenting path or alternating cycle
problem, the above results play a central role.

To decide the blossom-binding relation, we need to find blossoms with both the
stem and one intermediate edge prescribed. A first step towards this will be to
find augmenting paths crossing a given edge. We will see later that this is
NP-complete in general (\cref{prescribed-augmenting-npc}). However, the problem
is tractable in the absence of alternating cycles (plus a condition which will
always be satisfied in the various reductions to augmenting paths used in this
paper).

\begin{lemma}
  \label{lemma-prescribed}
  Let $G = (V,E)$ be a graph and $M$ be a matching of $G$. Suppose that:
  \begin{itemize}
  \item there are \emph{no alternating cycles} for $M$ -- equivalently, $M$ is
    the unique perfect matching of the subgraph induced by the vertices matched
    by $M$;
  \item there are exactly two unmatched vertices.
  \end{itemize}
  Then the existence of an augmenting path for $M$ \emph{crossing a prescribed
    matching edge $e \in M$} can be reduced to the existence of a perfect
  matching; and such a path can be found in linear time.
\end{lemma}

\begin{proof}
  Let $u,v \in V$ be the unmatched vertices. If there is an augmenting path for
  $M$ in $G$, its endpoints must be $u$ and $v$, and this is equivalent to the
  existence of a perfect matching in $G$. Let $e = (a,b)$, $G' = (V, E \setminus
  \{e\})$ and $M' = M \setminus \{e\}$.

  Suppose $G'$ admits a perfect matching $M''$. Then the symmetric difference
  $M' \triangle M''$ consists of two vertex-disjoint alternating paths for $M'$
  whose endpoints are $\{u, v, a, b\}$, by Berge's lemma for paths
  (\cref{berge-augmenting}); indeed, our assumptions prevent the existence of
  alternating cycles for $M$, and therefore for $M' \subset M$ as well.

  We claim that these paths either go from $u$ to $a$ and $b$ to $v$, or from
  $u$ to $b$ and $a$ to $v$. Otherwise, there would be an alternating path from
  $a$ to $b$ for $M'$ in $G'$, and together with $(a,b) = e \in M$, this would
  give us an alternating cycle for $M$ in $G$.

  In both cases, let us join the two paths together by adding $e$. We get a path
  starting with $u$, ending with $v$, crossing $e$ and alternating for $M$ in
  $G$. Conversely, from such a path, one can get a perfect matching in $G'$.

  For the linear time complexity, we exploit the fact that we already have at
  our disposal a matching $M'$ of $G'$ which leaves only 4 vertices unmatched. A
  perfect matching can then be found as follows: find a first augmenting path
  $P$ for $M'$ in linear time, and then a second one $P'$ for $M' \triangle P$,
  both steps being done in linear time by \cref{augmenting-path-linear}. If both
  augmenting paths exist, then $M \triangle P \triangle P'$ is a perfect
  matching, and conversely, if $G'$ admits a perfect matching, then the
  procedure succeeds in finding some $P$ and $P'$. (This does not mean that $P$
  and $P'$ are the same as the paths in the previous part of the proof, since
  they may not be vertex-disjoint.)
\end{proof}

\begin{theorem}
  Let $M$ be the unique perfect matching of a graph $G$, $u \in V$ and
  $e,f \in M$. The blossom-binding relations $u \twoheadrightarrow e$ and $e
  \rightarrow f$ can both be decided in linear time.
\end{theorem}
\begin{proof}
  For $e = (v,w)$, $e \rightarrow f \Longleftrightarrow v \twoheadrightarrow f
  \text{ or } w \twoheadrightarrow f$, so it suffices to treat the case $u
  \twoheadrightarrow e$. This is done by a reduction to the previous lemma.

  We build a graph $G'$ by adding two new vertices $s, t$ to $G$; the neighbors
  of $s$ in $G'$ are exactly those of $u$ in $G$, and the same goes for $t$. If
  $P$ is an augmenting path in $G'$, then by replacing both endpoints $s$ and
  $t$ by $u$, we get a blossom in $G$ whose root is $u$; the converse also
  holds.
\end{proof}

\begin{corollary}
  The kingdom ordering of a unique perfect matching can be decided in time
  $O(n^2 m)$ for a graph with $n$ vertices and $m$ edges.
\end{corollary}
\begin{proof}
  The previous result allows us to compute the blossom-binding relation between
  all pairs of matching edges in time $O(n^2 m)$ (there are $n/2$ matching
  edges). The conclusion follows from \cref{bellin-thm-graph} and the $O(n^3)$
  Floyd--Warshall algorithm for transitive closure.
\end{proof}

\section{A few remarks on properly colored paths and cycles}
\label{sec:remarks}

In this section, we mostly focus on recalling known result on edge-colored
graphs and provide some minor improvements which will be useful in the remainder
of the paper. Novelties such as the edge-colored line graph and its applications
are to be found in the next sections, not here.

As usual, there are both an efficient algorithm for finding properly colored
paths, and an associated \enquote{structure from acyclicity} theorem.

\begin{theorem}
  \label{pc-path-linear}
  Let $G$ be an edge-colored graph and $u,v$ be two vertices in $G$. A properly
  colored path between $u$ and $v$ can be found in linear time.
\end{theorem}

\begin{definition}
  In an edge-colored graph, the \emph{chromatic degree} of a vertex is the
  number of different colors among its incident edges.
\end{definition}

\begin{theorem}[Yeo \cite{yeo_note_1997}]
  \label{yeo}
  Let $G$ be an edge-colored graph with no properly colored cycles. Note $V$ for
  the vertex set of $G$. Then, $G$ has a \emph{color-separating vertex} $u \in
  V$: for each connected component $C$ in $G[V \setminus \{u\}]$, all edges
  between $C$ and $u$ have the same color. In particular, if $u$ has chromatic
  degree $\geq 2$, then it is a cut vertex.
\end{theorem}

The algorithmic part (\cref{pc-path-linear}) was actually first proven in
Szeider's paper on forbidden transitions~\cite{szeider_finding_2003}, by
reduction to augmenting paths. A more general presentation of this reduction,
parameterized by a choice of \enquote{P-gadget}, was later given in
\cite{gutin_properly_2009} (see also the exposition in \cite[Section
16.4]{bang-jensen_digraphs._2009}).

There is a much simpler reduction to matchings in the case of
\emph{2-edge-colored} graphs, i.e. edge-colored graphs using only two colors. It
seems to have been first published in~\cite{manoussakis_alternating_1995}, where
it is attributed to Edmonds. The existence of this reduction also makes the
2-edge-colored case of Yeo's theorem above immediately deducible from Kotzig's
theorem (\cref{kotzig}); that said, this case was first proved
in~\cite{grossman_alternating_1983} without using Kotzig's theorem.

Our two small contributions here are the following:
\begin{itemize}
\item To our knowledge, no linear time algorithm for finding properly colored
  \emph{cycles} appears in the literature. (A slightly worse than linear
  algorithm is proposed in \cite[Corollary~16.4.3]{bang-jensen_digraphs._2009}.)
  We repair this omission in this section, by a mostly straightforward
  adaptation of the reduction for paths in~\cite{gutin_properly_2009} -- the
  main subtlety being that we need to tweak the definition of
  \enquote{P-gadgets} for our purpose.
\item We generalize the reduction for the 2-edge-colored case to cover
  edge-colored graphs with chromatic degree $\leq 2$, and unify it with the
  well-known correspondence between bipartite matchings and directed graphs.
  This is achieved by extending it to a combinatorial bijection between graphs
  equipped with perfect matchings and a new object we call \enquote{locally
    2-colored graphs}.
\end{itemize}

\subsection{Finding properly colored cycles in linear time}
\label{sec:pc-cycle-linear}

\begin{definition}
  \label{def:p-gadget}
  A \emph{P-gadget} on the vertices $V$ is a graph $G = (V',E)$ equipped with a
  \emph{unique} perfect matching $M$ such that $V \subseteq V'$ and, for each
  nonempty $U \subseteq V$, $G[V' \setminus U]$ has at least one perfect
  matching if and only if $|U| = 2$.
\end{definition}

This differs slightly from the definition in \cite{gutin_properly_2009}:
\begin{itemize}
\item $M$ does not merely exist, but is included as part of the data (this is a
  minor detail);
\item more importantly, we require this perfect matching $M$ to be unique.
\end{itemize}
This last condition will help us for cycle existence problems.

\begin{remark}
  Szeider proved the equivalence between Kotzig's and Yeo's theorems
  (theorems~\ref{kotzig} and~\ref{yeo}) in~\cite{szeider_theorems_2004} using a
  particular P-gadget; the proof adapts using any P-gadget according to our
  definition, making use of this uniqueness condition.
\end{remark}

\begin{lemma}
  \label{p-gadget-existence}
  For every finite vertex set $V$, there exists a P-gadget on $V$ with size
  $O(|V|)$ and it can be constructed in time $O(|V|)$.
\end{lemma}
\begin{proof}
  Three possible constructions are given in~\cite[Subsection
  16.4.1]{bang-jensen_digraphs._2009} and one can check that they fit our
  altered definition of a P-gadget.
\end{proof}

\begin{theorem}
  \label{compatible-cycles}
  A properly colored cycle in an edge-colored graph can be found in linear time.
\end{theorem}
\begin{proof}
  We proceed by a linear-time reduction to the problem of finding an alternating
  cycle and apply \cref{augmenting-path-linear}.

  Let $G = (V, E)$ an edge-colored graph. Without loss of generality, we can
  assume that $G$ has no isolated vertices.

  For each $v \in V$, let $c_1, \ldots, c_k$ be the colors used by the edges
  incident to $v$, $k$ being the chromatic degree of $v$. Introduce fresh
  vertices $v_{c_1}, \ldots, v_{c_k}$ and construct a P-gadget $(G_v, M_v)$ on
  $\{v_{c_1}, \ldots, v_{c_k}\}$.

  Build a graph $G'$ by, first, taking the disjoint union of the $G_v$ for $v
  \in V$ and then, for each original edge $e = (u,v) \in E$ with color $c$,
  adding an edge between $u_c$ and $v_c$ in $G'$. The target of the reduction is
  this graph $G'$ equipped with the perfect matching $M' = \bigcup_v M_v$.
  Thanks to \cref{p-gadget-existence}, this reduction is indeed linear-time.

  We claim that this perfect matching is unique if and only if $G$ has no
  properly colored cycle. First, notice that the P-gadgets are disconnected from
  each other by the removal of $E'$ from $G'$. Thus, from the uniqueness
  condition in our definition of P-gadgets, it follows that $M'$ is the only
  perfect matching in $G'$ which does not intersect $E'$.

  Suppose $M'' \neq M'$ is a perfect matching. For each $v \in V$, let $V_v$ be
  the vertices of $G_v$, and $U_v$ be those covered by an edge in $M'' \cap E'$.
  $U_v \subseteq \{v_{c_1}, \ldots, v_{c_k}\}$ and, since $M''$ induces a
  perfect matching on $G_v[V_v \setminus U_v]$, $|U_v| \in \{0,2\}$. $M'' \cap
  E'$ therefore lifts to a non-empty set $C \subseteq E$ of edges in $G$ such
  that every vertex of $G$ is incident either to no edge, or to two edges with
  different colors, of $C$. In other words, $C$ is the edge set of a disjoint
  union of properly colored cycles.

  Thus, by finding an alternating cycle $C'$ for $M'$ in $G'$, we can obtain
  such an $M''$ as the symmetric difference $C' \triangle M'$, which allows us
  to retreive a properly colored cycle in~$G$.

  Conversely, properly colored cycles in $G$ map to alternating cycles for $M'$,
  up to a choice, for each visited vertex $v$, of a perfect matching of $G_v[V_v
  \setminus \{v_c, v_{c'}\}]$ where $c$ and $c'$ are the colors of the two edges
  of the cycle incident to $v$.
\end{proof}

\subsection{Locally 2-colored graphs}
\label{sec:locally}

As mentioned before, we now introduce a new object and put it in bijection with
graphs equipped with perfect matchings. This can be thought of as a new way to
see the latter, whose benefit is to make some reductions to matchings more
obvious.

We write $\partial(u)$ for the set of incident edges of a vertex $u$.

\begin{definition}
  \label{def:locally}
  A \emph{locally 2-colored graph} is a graph in which each vertex $u$ comes
  with a partition $\partial(u) = R_u \sqcup B_u$ (\enquote{red} and
  \enquote{blue} edges) of its incident edges. We allow not only simple graphs,
  but also multigraphs, under the condition that any pair of parallel edges has
  at least one endpoint where they differ in color.

  A path is said to be \emph{compatible} if, for every intermediate vertex $u$,
  one of its incident edges in the path is in $R_u$ and the other is in $B_u$.
  Compatible cycles are defined analogously.
\end{definition}

This generalizes properly colored paths in 2-edge-colored graphs, as well as in
any edge-colored graph with chromatic degree bounded by 2 -- the latter case
will arise in the next section. Local 2-colorings are in turn a particular case
of forbidden transitions (\cref{def:transition}), see \cref{remark-ft-loc}.

\begin{theorem}
  \label{bijection-perfect-matching}
  There is a bijection (modulo isomorphism) between graphs equipped with a
  perfect matching and locally 2-colored graphs, through which alternating paths
  (resp.\ cycles) correspond to compatible paths (resp.\ cycles).
\end{theorem}

\begin{proof}
  Let $G = (V, E, R, B)$ be a locally 2-colored graph; we will build a graph $G'
  = (V',E')$ with a perfect matching $M$ from $G$.

  For each vertex $u \in V$, add two new vertices $u_R, u_B$ to $V'$ and a edge
  $e_u = (u_R, u_B)$ to $E'$. Let $M = \{e_u \mid u \in V\}$. Then, for each
  edge $e = (u, v)$, add an edge $(u_X, v_Y)$ to $E'$, where $X, Y \in \{R,B\}$
  are determined uniquely by the condition $e \in X_u \cap Y_v$. The restriction
  on parallel edges in the definition of locally 2-colored (multi)graphs ensures
  that the resulting graph $G'$ is simple.

  The inverse bijection is given by contracting the edges in a perfect matching
  into vertices. As for the correspondence between compatible and alternating
  paths, it holds for the same reason as in the 2-edge-colored
  case~\cite[Lemma~1.1]{manoussakis_alternating_1995}.
\end{proof}

An example of this bijection is shown in \cref{fig:bijection-pm}.

\begin{figure}
  \centering
  \begin{subfigure}{7.3cm}
    \centering
    \begin{tikzpicture}
      \node[vertex] (w) at (0,2) {};
      \node[vertex] (x) at (4,2) {};
      \node[vertex] (y) at (0,0) {};
      \node[vertex] (z) at (4,0) {};
      \node[vertex] (o) at (2,1) {};
      \draw[thick] (x) -- (z);
      \draw[amber,thick] (x) -- (o);
      \draw[amber,thick] (z) -- (o);
      \draw[lavenderindigo,thick] (w) -- (o);
      \draw[lavenderindigo,thick] (y) -- (o);
    \end{tikzpicture}
    \caption{A graph with chromatic degree $\leq 2$}
    \label{fig:semidirect}
  \end{subfigure}~~~~~\begin{subfigure}{7.3cm}
    \centering
    \begin{tikzpicture}
      \node[vertex] (w) at (0,2) {};
      \node[vertex] (x) at (4,2) {};
      \node[vertex] (y) at (0,0) {};
      \node[vertex] (z) at (4,0) {};
      \node[vertex] (o) at (2,1) {};

      \node[vertex] (w') at (1,5) {};
      \node[vertex] (x') at (5,5) {};
      \node[vertex] (y') at (1,3) {};
      \node[vertex] (z') at (5,3) {};
      \node[vertex] (o') at (3,4) {};

      \draw[matching edge] (x) -- (x');
      \draw[matching edge] (y) -- (y');
      \draw[matching edge] (z) -- (z');
      \draw[matching edge] (w) -- (w');
      \draw[matching edge] (o) -- (o');

      \draw[non matching edge] (x) -- (z);
      \draw[non matching edge] (x') -- (o');
      \draw[non matching edge] (z') -- (o');
      \draw[non matching edge] (w) -- (o);
      \draw[non matching edge] (y) -- (o);
    \end{tikzpicture}
    \caption{The corresponding graph with perfect matching (the matching edges
      are drawn in thick blue)}
  \end{subfigure}
  \caption{The bijection between locally 2-colored graphs and graphs equipped with
    a perfect matching}
\label{fig:bijection-pm}
\end{figure}
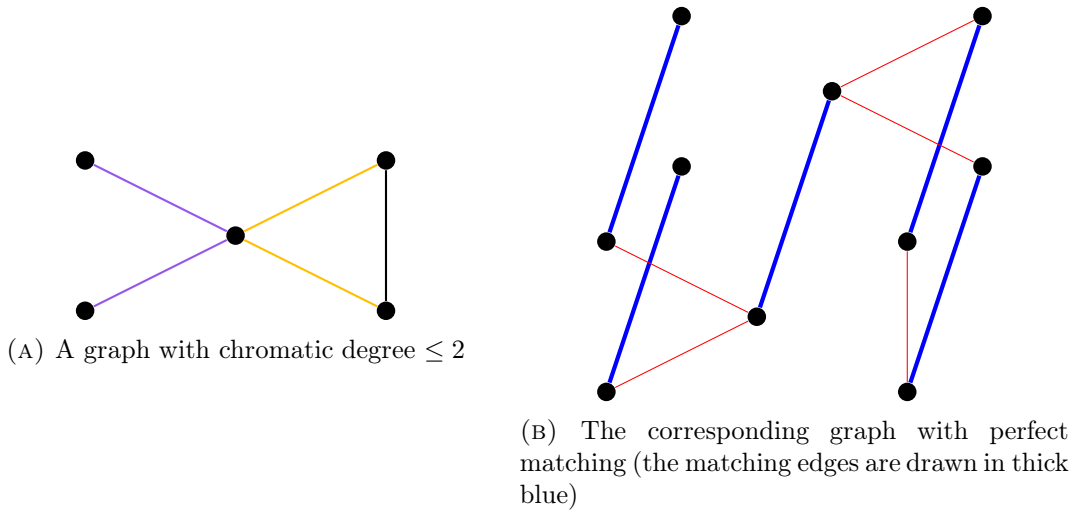

\begin{corollary}
  \label{compatible-paths-2-colored}
  In a locally 2-colored graph, finding a compatible path between two given
  vertices, or a compatible cycle, can be done in linear time.

  Furthermore, if the graph has no compatible cycle, then a compatible path
  joining two given endpoints and visiting a third given intermediate vertex can
  be found in linear time.
\end{corollary}
\begin{proof}
  Let $G$ be a locally 2-colored graph. The corresponding graph with perfect
  matching $(G',M)$ can be computed in linear time.

  If we are looking for a compatible cycle in $G$, it suffices to find an
  alternating cycle for $M$ in~$G'$. To find a path, suppose that $e, f \in M$
  are the edges of $G'$ corresponding to the required endpoints in $G$; we add
  two auxiliary unmatched vertices $u,v$ to $G'$, join $u$ to both endpoints of
  $e$, and similarly for $v$ and $f$, so that augmenting paths will correspond
  to compatible paths. In both cases, the result follows from
  \cref{augmenting-path-linear}. For the path visiting a prescribed vertex, we
  apply \cref{lemma-prescribed}.
\end{proof}

Using an existing result on 2-edge-colored graphs, we can now show that the
acyclicity assumption in the above corollary and in \cref{lemma-prescribed} is
necessary.

\begin{theorem}[\cite{chou_paths_1994}]
  For a 2-edge-colored graph $G$ and three vertices $s$, $t$ and $u$ in $G$,
  the existence of a properly colored path between $s$ and $t$ containing $u$ is
  an NP-complete problem.
\end{theorem}
\begin{corollary}
  \label{prescribed-augmenting-npc}
  Finding an augmenting path crossing a prescribed matching edge is an
  NP-complete problem, even when there are only two unmatched vertices.
\end{corollary}
\begin{proof}
  NP-hardness is established by the same reduction as the previous corollary,
  and the problem is in NP (guess the path non-deterministically).
\end{proof}

Beyond its algorithmic applications, this combinatorial bijection also allows us
to recover a well-known equivalence. A \emph{directed} graph induces a local
2-coloring on the underlying undirected (multi)graph, which distinguishes
incoming and outgoing arcs at each vertex; directed paths then correspond to
compatible paths. On the other hand, finding an augmenting path for a matching
in a \emph{bipartite} graph amounts to traversing a directed graph; indeed, this
is implicitly what happens when applying the Ford--Fulkerson flow algorithm to
find a maximum matching, and it is also why historically bipartite maximum
matchings were solved before the general case.

\begin{proposition}
  The previous bijection sends \emph{bipartite} graphs equipped with a perfect
  matching on the locally 2-colored encodings of \emph{directed} graphs and vice
  versa.
\end{proposition}
\begin{proof}
  In a locally 2-colored graph coming from a directed graph, each edge has two
  different colors since it is incoming for one of its endpoints and outgoing
  for the other one. Conversely, any such locally 2-colored graph can be
  realized as a directed graph: simply orient each edge from its blue side to
  its red side.

  Now, this means that in the corresponding graph with perfect matching, every
  edge is between $V_R = \{u_R \mid u \in V\}$ and $V_B = \{v_B \mid v \in V\}$
  (reusing the notations from a previous proof). Therefore, the graph is
  bipartite, with the partition of the vertices being $V' = V_R \sqcup V_B$.
\end{proof}

\begin{remark}
  One could actually develop in a straightforward way, using the machinery of
  P-gadgets (previous subsection), a theory of \enquote{locally $k$-colored
    graphs} for $k \geq 2$, with linear time algorithms and a general
  formulation of the \enquote{local Yeo's theorem}. We will not detail this
  generalization since it will not be needed in the rest of the paper.
\end{remark}

\section{Graphs with forbidden transitions}


As mentioned in the introduction, we will turn our attention to a very general
notion of local constraints on paths and trails.

\begin{definition}[\cite{szeider_finding_2003}]
  \label{def:transition}
  Let $G = (V,E)$ be a graph. A \emph{transition graph} for a vertex $v \in
  V$ is a graph whose vertices are the edges incident to $v$ : $T(v) =
  (\partial(v), E_v)$. A \emph{transition system} on $G$ is a family $T =
  (T(v))_{v \in V}$ of transition graphs. A graph equipped with a transition
  system is called a \emph{graph with forbidden transitions}.

  A path (resp.\ trail) $v_1, e_1, v_2 \ldots, e_{k-1}, v_k$ is said to be
  \emph{compatible} if for $i = 1, \ldots, k-1$, $e_i$ and $e_{i+1}$ are
  adjacent in $T(v_{i+1})$. For a \emph{cycle} (resp.\ \emph{closed} trail), we
  also require $e_{k-1}$ and $e_1$ to be adjacent in $T(v_1) = T(v_k)$.

  (That is, the edges of the transition graphs specify the \emph{allowed}
  transitions, i.e.\ the pairs of edges which may occur consecutively in a
  compatible path or trail.)
\end{definition}

\begin{remark}
  \label{remark-ft-loc}
  An edge-coloring of a graph induces a transition system made of complete
  multipartite graphs: for each vertex $v$, two edges of $\partial(v)$ are
  adjacent in $T(v)$ if and only if they have different colors. Properly colored
  paths (resp.\ trails) are exactly the paths (resp.\ trails) compatible with
  this transition system.

  For a similar reason, locally 2-colored graphs (\cref{def:locally}) correspond
  exactly\footnote{Modulo the fact that in a locally 2-colored graph, a pair of
    vertices may have two parallel edges joining them, if they are colored
    differently on one endpoint.} to graphs with forbidden transitions whose
  transition graphs are all complete bipartite, the notion of compatible path /
  cycle / (closed) trail being the same.
\end{remark}

\subsection{The edge-colored line graph}

We now introduce a key construction of this paper.


\begin{definition}
  \label{def:lec}
  Let $G = (V, E)$ be a graph and $T$ be a transition system on $G$. The
  \emph{EC-line graph} $L_{EC}(G, T)$ is formed by taking the line graph of $G$,
  coloring its edges so that the clique corresponding to $v$ is given the color
  $v$ (using the vertices of $G$ as the set of colors), and deleting the edges
  corresponding to forbidden transitions.

  Formally, $L_{EC}(G, T)$ is defined as the graph with vertex set $E$ and edge
  set $F = \bigsqcup_{v \in V} T(v)$, equipped with an edge coloring $c : E' \to
  V$ with values in $V$: for $f \in F$, $c(f)$ is the unique vertex such that $f
  \in T(c(f))$.
\end{definition}



Consider a vertex in the line graph, corresponding to an edge $e = (u,v)$ in the
original graph. Its neighbors are the edges which share a vertex with $e$.
Equipping the line graph with an edge coloring allows us to distinguish between
the neighbors of $e$ incident to $u$ and those incident to $v$. The following
proposition encapsulates the usefulness of this additional information.

\begin{proposition}
  \label{lec-magic}
  The compatible \emph{cycles} (resp.\ \emph{closed trails}) of a graph
  with forbidden transitions correspond bijectively to \emph{rainbow} (resp.\
  \emph{properly colored}) cycles in its EC-line graph.
\end{proposition}
\begin{proof}
  What happens here is fairly intuitive. Let us however pinpoint a crucial role
  played by the edge coloring: it excludes monochromatic sub-paths of length $>
  1$ in transition graphs, which could allow forbidden transitions to be taken.
  For instance, suppose we have a vertex with three incident edges $e, f, g$
  such that the allowed transitions are $e \leftrightarrow f$ and $f
  \leftrightarrow g$; then a path in the EC-line graph containing the sub-path
  $e \to f \to g$ would translate, in the original graph, to a path where
  $e$ and $g$ occur consecutively.
  
  This explains why we want properly colored cycles -- rainbow cycles being, in
  particular, properly colored. As for the distinction between compatible cycles
  and closed trails, notice that the global condition of non-repetition of
  vertices translates in the EC-line graph to requiring a rainbow cycle, whereas
  a repeated edge in the original graph would correspond to a repeated
  vertex in the EC-line graph.
\end{proof}

The situation for paths and trails is analogous to the above, a few precisions
being necessary to ensure bijectivity.

\begin{proposition}
  Let $G$ be a graph with transition system $T$, and $s,t$ be two distinct
  vertices of $G$.

  The compatible \emph{paths} between $s$ and $t$ correspond bijectively to
  \emph{rainbow} paths in $L_{EC}(G,T)$ between some vertex of $\partial(s)$
  (identified with a subset of vertices in $L_{EC}(G,T)$) and some vertex of
  $\partial(t)$ which do not cross any edge with color $s$ or $t$.

  Similarly, the compatible \emph{trails} between $s$ and $t$ where neither $s$
  nor $t$ appear as intermediate vertices correspond bijectively to
  \emph{properly colored} paths in $L_{EC}(G,T)$ between some vertex of
  $\partial(s)$ and some vertex of $\partial(t)$ which do not cross any edge
  with color $s$ or~$t$.
\end{proposition}

From now on until the end of this section, we will focus mostly on the
correspondence between compatible trails and properly colored paths. To apply
the results of \cref{sec:locally}, we remark that:

\begin{proposition}
  \label{ec-2col}
  EC-line graphs have chromatic degree $\leq 2$, and can therefore be seen as
  locally 2-colored paths.
\end{proposition}

Composing the EC-line graph construction with the bijection from
\cref{bijection-perfect-matching}, we end up with a direct reduction to
matchings. It turns out to be simple and clean, because we exploit the low
chromatic degree; otherwise, we would have had to use the reduction based on
P-gadgets (\cref{sec:pc-cycle-linear}) instead of the bijection, which would
have given a messier result.

\begin{definition}
  \label{def:lpm}
  The \emph{PM-line graph} $L_{PM}(G,T)$ is defined as the graph:
  \begin{itemize}
  \item with vertex set $\{ u_e \mid e \in E,\, u \text{ is an endpoint of } e \}$;
  \item with edge set $M \sqcup E'$, where
    \begin{itemize}
    \item $M = \{ (u_e, v_e) \mid e = (u,v) \in E \}$;
    \item $E' = \{ (u_e, u_f) \mid u \in V,\, e,f \in \partial(u)
      \text{ are adjacent in } T(u) \}$;
    \end{itemize}
  \item equipped with the perfect matching $M$.
  \end{itemize}
\end{definition}


\begin{proposition}
  Compatible closed trails in a graph with forbidden transitions correspond
  bijectively to alternating cycles of its PM-line graph.
\end{proposition}
\begin{proof}
  By combining \cref{bijection-perfect-matching} with the relevant half of
  \cref{lec-magic}.
\end{proof}

Naturally, a similar correspondence holds for compatible trails between two
given vertices. However, a disadvantage of $L_{PM}$ with respect to $L_{EC}$ is
that the former does not provide a counterpart to \emph{paths} avoiding
forbidden transitions.

\begin{remark}
  This PM-line graph construction appears implicitly in Retoré's
  work~\cite{retore_handsome_2003} on proof nets. This remark was indeed the
  starting point of the present work; we refer the reader back to the
  introduction for further discussion of this point.
\end{remark}

\subsection{Algorithms for trails avoiding forbidden transitions}

With the tools now at our disposal, we are in a position to give an algorithm
for finding compatible trails.

\begin{theorem}
  \label{compatible-trail}
  Deciding whether, in a graph $G$ with transition system $T$, two given
  vertices can be joined by a compatible trail, and computing such a trail, can
  be done in linear time (in the size of the input, including $T$).

  Similarly, a compatible closed trail in $G$ can be found in linear time.
\end{theorem}

To be more precise about what we mean by \enquote{size of the input}, let us say
we represent each transition graph $T(v)$ using an adjacency list. Then the size
$|T(v)|$ of this representation is the number of edges in $T(v)$ and can go up to
$\mathrm{deg}(v)^2$. Thus, the total size of the input is $\Theta(|T|)$, where $|T| =
\sum_v |T(v)| \leq \sum_v\mathrm{deg}(v)^2$ (we have $|E| = O(|T|)$ and, assuming there
is no isolated vertex, $|V| = O(|T|)$).

For an edge-colored graph, writing down the explicit transition graphs may
result in a non-linear increase in size. A linear-time algorithm for finding a
properly colored trail between two given vertices of an edge-colored graph is
given in \cite[Corollary 2.3]{abouelaoualim_paths_2008}.

\begin{proof}[Proof of \cref{compatible-trail}]
  Thanks to \cref{ec-2col}, it suffices to apply the algorithm
  of \cref{compatible-paths-2-colored} to $L_{EC}(G,T)$, the running time being
  therefore linear in the size of $L_{EC}(G,T)$. Since the edge set of
  $L_{EC}(G,T)$ is $\bigcup_{v \in V} T(v)$, the size is also $\Theta(|T|)$.
\end{proof}

\begin{theorem}
  For a graph $G$, a transition system $T$ on $G$, two vertices $s,t$ and
  an edge $e$ in $G$, the existence of a compatible trail from $s$ to
  $t$ crossing $e$:
  \begin{itemize}
  \item is a NP-complete problem in general;
  \item can be decided in linear time in the size of the input, including $T$,
    when $G$ contains no compatible closed trail.
  \end{itemize}
  In the latter case, the linear time algorithm also solves the corresponding
  search problem.
\end{theorem}
\begin{proof}
  Finding an augmenting path for a matching crossing a given edge is an
  instance of this problem: indeed, an \enquote{augmenting trail} cannot visit
  the same vertex $u$ twice, since it would have to go through the unique
  matching edge incident to $u$ twice. Thus, the NP-completeness result from
  \cref{prescribed-augmenting-npc} applies.

  Suppose now that $G$ contains no compatible closed trail. Using the EC-line
  graph construction, the problem reduces to finding a path in $L_{EC}(G,T)$
  between $\partial(u)$ and $\partial(v)$ visiting $e$, knowing that
  $L_{EC}(G,T)$ has no properly colored cycle. This can be done in linear time
  by the second half of \cref{compatible-paths-2-colored}, thanks again to
  \cref{ec-2col}.
\end{proof}

\subsection{A structural theorem}

As for alternating cycles in perfect matchings, the absence of closed trails
avoiding forbidden transitions leads to the existence of a bridge -- provided we
add another assumption.

\begin{theorem}
  \label{forbidden-transition-bridge}
  Let $G$ be a graph with transition system $T$, with at least one edge.
  If, for all vertices $v$ in $G$, the transition graph $T(v)$ is
  \emph{connected}, and $G$ has no compatible closed trail, then $G$ has a
  bridge.
\end{theorem}

The connectedness assumption is not too restrictive: if $T(v)$ has multiple
connected components, then $v$ can be split into multiple vertices corresponding
to those connected components without changing the reachability by compatible
trails. However, the bridges of the graph after this transformation are not
necessarily bridges of the original graph.

\begin{proof}
  Again, our proof uses the EC-line graph structure (\cref{def:lec}). Without
  loss of generality, we assume $G$ to be connected (as a graph, without
  taking the transition system into consideration). Since $G$ has at least one
  edge, $L_{EC}(G,T)$ is not the empty graph; since the transition graphs are
  connected, it can be seen $L_{EC}(G,T)$ is connected as well; and the absence
  of compatible closed trail in $G$ means that it has no properly colored cycle.
  Therefore, by Yeo's theorem (\cref{yeo}), there is a color-separating vertex
  $e$ in $L_{EC}(G,T)$ -- the notation reminding us that it corresponds to an
  edge of $G$.

  If $e$ has chromatic degree~0 (resp.~1) in $L_{EC}(G,T)$, it means that both
  its endpoints are (resp. one of its endpoints is) a degree~1 vertex in $G$. In
  both cases, $e$ is a bridge of $G$ and we are done.

  Else, $e$ has chromatic degree~2, and since it is color-separating, it is a
  cut-vertex of $L_{EC}(G,T)$. But we have\footnote{Taking a vertex-induced
    subgraph on the left of the $\simeq$ sign, and removing an edge while
    leaving the vertices intact on the right.} $L_{EC}(G,T) \setminus \{e\}
  \simeq L_{EC}(G \setminus \{e\},T)$, and so, again using the connectedness of
  the transition graphs, since $L_{EC}(G,T) \setminus \{e\}$ is disconnected, $G
  \setminus \{e\}$ is disconnected. In other words, $e$ is a bridge of $G$.

  The same argument can be replayed by applying Kotzig's theorem (\cref{kotzig})
  to $L_{PM}(G,T)$ (\cref{def:lec}).
\end{proof}

As a corollary, we obtain a new proof of the \enquote{structure from acyclicity}
property for properly colored trails. (The original
proof~\cite{abouelaoualim_paths_2008} applies Yeo's theorem to a construction
which is rather different from our EC-line graph and which does not generalize
to forbidden transitions.)

\begin{corollary}[{\cite[Theorem~2.4]{abouelaoualim_paths_2008}}]
  \label{edge-colored-kotzig}
  Let $G$ be an edge-colored graph such that every vertex of $G$ is incident
  with at least two differently colored edges. Then, if $G$ does not have a
  properly colored closed trail, then $G$ has a bridge.
\end{corollary}

\begin{proof}
  Let $T$ be the transition system induced by the edge coloring of $G$, and $v$
  be any vertex of~$G$. $T(v)$ is a complete $k$-partite graph, $k$ being the
  chromatic degree of $v$. Since $v$ is incident to at least two differently
  colored edges, $k \geq 2$ making $T(v)$ connected. Thus, the previous theorem
  can be applied in this case.
\end{proof}

To conclude this section, we show that our theorem on closed trails avoiding
forbidden transitions actually \emph{implies} to Kotzig's theorem
(\cref{kotzig}) -- and therefore, is equivalent to both Kotzig's and Yeo's
theorems.

\begin{proof}[Proof of Kotzig's theorem from \cref{forbidden-transition-bridge}]
  We use a proof by contradiction. Let $G = (V,E)$ be a minimal counterexample,
  and $M$ be its unique perfect matching. It has no degree 1 vertex since then
  the incident matching edge would be a bridge. Thus, every vertex is incident
  to exactly one matching edge and at least one edge outside the matching.

  Since alternating paths a matching $M$ are the same as properly colored closed
  trails for the induced 2-edge-coloring, \cref{edge-colored-kotzig} applies to
  show $G$ has at least one bridge. The bridges of $G$ being outside the
  matching, after removing them, which disconnects the 2-edge-connected
  components of $G$, $M$ is still the unique perfect matching of $G$.

  Now, let $G' = (V',E')$ a 2-edge-connected component of $G$. Then $G'$ is
  strictly smaller than $G$, and since $G'$ is a connected component of a
  spanning subgraph of $G$ which contains $M$ as a perfect matching, $M \cap E'$
  is a perfect matching of $G'$. It is also unique, or else $M$ would not be
  unique in $G$. Thus, $G'$ is a counterexample to Kotzig's theorem,
  contradicting the minimality of $G$.
\end{proof}

\section{Finding rainbow paths}
\label{sec:rainbow-path}

The EC-line graph construction connects forbidden transitions and edge-colored
graphs. The previous section applied this connection to the former; here we are
interested in drawing the consequences for the latter.

A \emph{color class} in an edge-colored graph is the set of all edges with some
common color; it may also refer to the subgraph edge-induced by such a set. The
purpose of this section is to study the complexity of the following problem.

\begin{definition}
  The problem $\mathcal{A}$-\textsc{Rainbow Path} (for $\mathcal{A}$ a class of
  graphs) is defined as:
  \begin{itemize}
  \item Input: an edge-colored graph $G$, whose color classes all induce graphs
    belonging to $\mathcal{A}$ (up to isomorphism), and two vertices $s$ and $t$
    in $G$.
  \item Output: a rainbow path between $s$ and $t$ in $G$.
  \end{itemize}
  We will also use the same name to refer to the decision problem which asks for
  the existence of such a path.
\end{definition}

Using the notations from \cite{szeider_finding_2003}, we write
$\mathcal{A}^{\mathrm{ind}}$ for the closure of $\mathcal{A}$ under taking
vertex-induced subgraphs, and $K_2 + K_2$, $P_4$ and $L_4$ refer to the graphs
shown in \cref{fig:excluded}. We also write $K_2$ for the complete graph on 2
vertices.

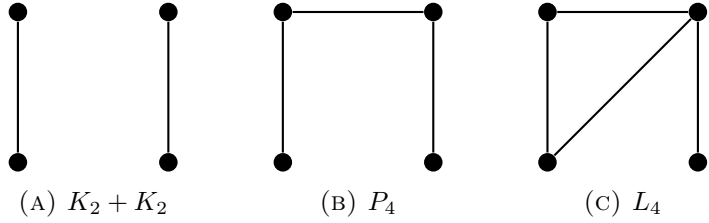
\begin{figure}
  \centering
  \begin{subfigure}{2.5cm}
    \centering
    \begin{tikzpicture}
      \node[vertex] (w) at (0,2) {};
      \node[vertex] (x) at (2,2) {};
      \node[vertex] (y) at (0,0) {};
      \node[vertex] (z) at (2,0) {};
      \draw[thick] (w) -- (y);
      \draw[thick] (x) -- (z);
    \end{tikzpicture}
    \caption{$K_2 + K_2$}
  \end{subfigure}~\qquad\qquad~\begin{subfigure}{2.5cm}
    \centering
    \begin{tikzpicture}
      \node[vertex] (w) at (0,2) {};
      \node[vertex] (x) at (2,2) {};
      \node[vertex] (y) at (0,0) {};
      \node[vertex] (z) at (2,0) {};
      \draw[thick] (w) -- (y);
      \draw[thick] (x) -- (z);
      \draw[thick] (w) -- (x);
    \end{tikzpicture}
    \caption{$P_4$}
    \label{subfig:p4}
  \end{subfigure}~\qquad\qquad~\begin{subfigure}{2.5cm}
    \centering
    \begin{tikzpicture}
      \node[vertex] (w) at (0,2) {};
      \node[vertex] (x) at (2,2) {};
      \node[vertex] (y) at (0,0) {};
      \node[vertex] (z) at (2,0) {};
      \draw[thick] (w) -- (y);
      \draw[thick] (x) -- (z);
      \draw[thick] (w) -- (x);
      \draw[thick] (y) -- (x);
    \end{tikzpicture}
    \caption{$L_4$}
  \end{subfigure}
  \caption{Excluded vertex-induced subgraphs}
  \label{fig:excluded}
\end{figure}

The algorithmic results of this section can be summarized as follows.

\begin{theorem}
  \label{dichotomy-thm}
  If $\mathcal{A}^{\mathrm{ind}}$ contains $K_2 + K_2$, $P_4$ or $L_4$, then
  the problem $\mathcal{A}$-\textsc{Rainbow Path} is \emph{NP-complete}.

  Else, every graph in $\mathcal{A}$ is the union of a \emph{complete
    multipartite} graph and of isolated vertices, and
  $\mathcal{A}$-\textsc{Rainbow Path} can be solved in \emph{linear} time.
\end{theorem}

The case singled out as tractable by this theorem deserves a name, which we take
from~\cite{alon_multicolored_1991}.

\begin{definition}
  \label{def:multidec}
  A \emph{multipartite decomposition} of a graph is an edge coloring whose color
  classes are all complete multipartite.
\end{definition}

Again, we have a \enquote{structure from acyclicity} result:
\begin{theorem}
  \label{rainbow-structure}
  Let $G$ be a rainbow acyclic graph whose edge coloring is a multipartite
  decomposition. Suppose $G$ has at least one edge. Then there exists a color
  class $H$ with vertex partition $U_1, \ldots, U_k$, such that removing the
  edges of $H$ disconnects $U_1, \ldots, U_k$ -- or in other words, such that
  for all $v \in U_i$ and $w \in U_j$ with $i \neq j$, $v$ and $w$ are not
  connected in $G \setminus H$.
\end{theorem}

The special case of \emph{bipartite} decompositions has been considered in a
number of works on combinatorics focusing on the minimum number of colors needed
for a bipartite decomposition of a given graph. For instance, a well-known
result is that all bipartite decompositions of the complete graph $K_n$ use at
least $n-1$ colors~\cite{graham_addressing_1971}.

As for rainbow \emph{paths and cycles} in bipartite decompositions, as mentioned
in the introduction, they were considered under the name \enquote{aggregates}
in~\cite[Chapter~2]{retore_reseaux_1993}, their study being motivated by linear
logic. A proof of \cref{rainbow-structure} for this special case, which does not
rely on another result such as Kotzig's or Yeo's theorems (respectively
theorems~\ref{kotzig} and~\ref{yeo}), was the main result of this chapter
\cite[Theorem~2.4]{retore_reseaux_1993}. Retoré would later adapt this into a
direct proof of Kotzig's theorem in \cite[Appendix]{retore_handsome_2003}.

\subsection{Hardness results on rainbow paths}

Let us establish the first half of \cref{dichotomy-thm}. In the following proofs of
NP-completeness, we only treat the NP-hardness part, since it is clear that the
problems belong to NP.

\begin{proposition}
  If $P_4 \in \mathcal{A}^{\mathrm{ind}}$ or $L_4 \in
  \mathcal{A}^{\mathrm{ind}}$, then $\mathcal{A}$-\textsc{Rainbow Path} is
  NP-complete.
\end{proposition}

\begin{proof}
  Let us call $\mathcal{A}$-\textsc{Compatible Path} the problem of finding a
  path avoiding forbidden transitions between two given vertices when the
  transition graphs are all in $\mathcal{A}$. There is a polynomial-time
  reduction from $\mathcal{A}$-\textsc{Compatible Path} to
  $\mathcal{A}$-\textsc{Rainbow Path} by \cref{lec-magic}. The former problem
  was shown to be NP-complete when $P_4 \in \mathcal{A}^{\mathrm{ind}}$ or $L_4
  \in \mathcal{A}^{\mathrm{ind}}$ in \cite[Theorem 1]{szeider_finding_2003}.
\end{proof}

\begin{lemma}
  $\{K_2, K_2 + K_2\}$-\textsc{Rainbow Path} is NP-complete.
\end{lemma}
\begin{proof}
  The NP-completeness of the general rainbow path problem is proved in
  \cite[Theorem 2.3]{chakraborty_hardness_2011} by a reduction from from
  3-\textsc{Sat}. Examining the proof reveals that the instances generated by
  the reduction are edge-colored graphs whose color classes are all isomorphic
  either to $K_2$ (a single edge) or $K_2 + K_2$. Thus, this proof actually
  provides a polynomial reduction from 3-\textsc{Sat} to $\{K_2, K_2 +
  K_2\}$-\textsc{Rainbow Path}.
\end{proof}

\begin{proposition}
  If $K_2 + K_2 \in \mathcal{A}^{\mathrm{ind}}$, then
  $\mathcal{A}$-\textsc{Rainbow Path} is NP-complete.
\end{proposition}

\begin{proof}
  We proceed by reduction from $\{K_2, K_2 + K_2\}$-\textsc{Rainbow Path}.

  Suppose $K_2 + K_2 \in \mathcal{A}^{\mathrm{ind}}$. We can then choose a graph
  $\Gamma \in \mathcal{A}$ with a set of four vertices $W$ such that the induced
  graph $\Gamma[W]$ is isomorphic to $K_2 + K_2$.

  Now, let $G$ be an instance of $\{K_2, K_2 + K_2\}$-\textsc{Rainbow Path}.
  Consider a color class $H$ in $G$ such that $H \simeq K_2 + K_2$. We build a
  new graph $G'$ by removing the edges of $H$, adding (a copy of) $\Gamma$ with
  the same color as $H$, and identifying the vertices of $H$ with (the copy of)
  $W$. This is equivalent to adding the vertices in $\Gamma \setminus W$ and
  adding some edges with the same color as $H$, with every new edge having at
  least one endpoint among the new vertices, and without removing anything. This
  modification does not affect the set of rainbow paths between $s$ and $t$, in
  particular, it does not change its emptiness or nonemptiness, since any
  rainbow path starting at $s$ and crossing a new edge would end up stuck in a
  new vertex.

  Thus, by induction, we can replace all color classes isomorphic to $K_2 + K_2$
  with copies of our gadget $\Gamma$, without affecting the existence of a
  rainbow path between $s$ and $t$. The same operation can be carried out to
  replace all the color classes consisting of single edges by copies of
  $\Gamma$: indeed, if $K_2 + K_2$ is a vertex-induced subgraph of $\Gamma$,
  then $K_2$ is as well. In the end, all color classes are isomorphic to
  $\Gamma$: we have constructed an instance of $\mathcal{A}$-\textsc{Rainbow
    Path}.

  All we have left to prove is that our reduction can be computed in polynomial
  time. This is clear if we remember that the size of $\Gamma$ is $O(1)$, since
  it depends only on $\mathcal{A}$ and not on the input.
\end{proof}

\subsection{Solving the tractable case}

The next 2 propositions cover the remaining half of \cref{dichotomy-thm}.

\begin{proposition}
  If neither $K_2 + K_2$, $P_4$ nor $L_4$ are in $\mathcal{A}^{\mathrm{ind}}$,
  then every graph in $\mathcal{A}$ is the union of a complete multipartite
  graph with isolated vertices.
\end{proposition}
\begin{proof}
  This is shown as part of the proof of \cite[Lemma 7]{szeider_finding_2003}.
\end{proof}

In fact, since edge-induced subgraphs cannot have isolated vertices, an instance
for $\mathcal{A}$-\textsc{Rainbow Path} for such a class $\mathcal{A}$ can only
have complete multipartite graphs as color classes, that is, it is necessarily a
multipartite decomposition.



\begin{proposition}
  \label{rainbow-path-linear}
  There are linear-time algorithms for finding a rainbow path or cycle in a
  multipartite decomposition.
\end{proposition}
\begin{proof}
  We treat here the case of paths; the case of cycles uses the same reduction.
  
  Let $G$ be a graph with a multipartite decomposition of its edges, $V$ be its
  set of vertices, and $H_1, \ldots, H_k$ be its color classes. Define $G'$ as
  the graph with vertices $V \sqcup W$, where $W = \{w_1, \ldots, w_k\}$ is a
  set of with one fresh vertex per color class, and with edges $(v, w_i)$ for
  all $v \in V$ and $i \in \{1, \ldots, k\}$ such that $v \in H_i$. Note that
  $G'$ is bipartite with partition $(V, W)$. Color the edges of $G$ in such a
  way that $(u, w_i)$ and $(v, w_j)$ have the same color if and only if $i = j$
  and $u$ is \emph{not} adjacent to $v$ in $H_i$, or equivalently, if $u$ and
  $v$ are in the same part of the vertex partition of $H_i$.

  We claim that \emph{properly colored} paths in $G'$ between vertices of $V$
  correspond bijectively to \emph{rainbow} paths for $G$. Let $s = v_1, w_1
  \ldots, w_{n-1}, v_n, = t$ be a properly colored path between $s \in V$ and $t
  \in V$ in $G'$, with $v_i \in V$ and $w_i \in W$ for all $i \in \{1, \ldots,
  n-1 \}$. Our choice of coloring, together with the fact that the path is
  properly colored, ensures that for all $i$, $(v_i, v_{i+1})$ is an edge in
  $H_i$. Thus, the path corresponds to a path $P$ in $G$. By definition, a path
  has no repeated vertices, so the $w_i$ -- and therefore the $H_i$ -- are
  distinct, which means that $P$ is actually a rainbow path. It is clear that
  this defines a bijection.

  Thus, our algorithm is to build $G'$ and then find a properly colored path in
  it. The time complexity is linear thanks to the following facts:
  \begin{itemize}
  \item the uncolored graph $G'$ can be constructed in linear time, since the
    number of edges to add between $V$ and $W$ is at most the sum of the numbers
    of vertices of each $H_i$, which is itself bounded by twice the number of
    edges of $G$;
  \item the edges of $G'$ can be colored in linear time -- in fact, this
    requires computing the vertex partition of a complete multipartite graph in
    linear time, which is non-trivial, see \cref{partition-linear};
  \item properly colored paths can be found in linear time (\cref{pc-path-linear}).
  \end{itemize}
\end{proof}

\begin{lemma}
  \label{partition-linear}
  The vertex partition of a complete multipartite graph can be computed in
  linear time.
\end{lemma}
\begin{proof}
  Recall that the class of \emph{cographs}~\cite{corneil_complement_1981} is the
  smallest class of graphs containing the one-vertex graph and closed under
  disjoint unions and complementation. Equivalently, a cograph is a graph which
  can be described by a \emph{cotree}: a rooted tree whose leaves are labeled
  with the vertices of the graph, and whose internal nodes are labeled with
  either $\land$ or $\lor$, such that two vertices are adjacent iff the lowest
  common ancestor of the corresponding leaves is a $\land$ node. If we require
  all paths from the root to the leaves to alternate between $\land$ and $\lor$,
  this makes the cotree for a cograph unique.

  Complete $k$-partite graph for are a particular case of cographs: for $k \geq
  2$ (resp.\ $k=1$) they are the cographs whose canonical cotree has depth~2
  (resp.~1) and whose root has label $\land$ (resp.~$\lor$). In this case, the
  immediate subtrees of the root describe the vertex partition. Fortunately,
  there are several algorithms for computing a cotree in linear time, e.g.\
  \cite{corneil_linear_1985,bretscher_simple_2008}.
\end{proof}

Using the same reduction as \cref{rainbow-path-linear}, we can also prove
\enquote{structure from acyclicity} for rainbow acyclic multipartite
decompositions. The main ingredient in our proof will be Yeo's theorem
(\cref{yeo}).

\begin{proof}[Proof of \cref{rainbow-structure}]
  We proceed by strong induction on the size of $G$. Let $G'$ be the
  corresponding edge-colored graph with only stars as color classes, as
  constructed in the algorithm of \cref{rainbow-path-linear}; we reuse the
  notations from its proof, writing $V$ and $W$ respectively for the vertices of
  $V$ and the additional vertices in $G'$ corresponding to color classes.

  Since $G$ is rainbow acyclic, $G'$ has no properly colored cycle. By Yeo's
  theorem, $G'$ has a color-separating vertex $u \in V \sqcup W$. It remains to
  do a case analysis on this vertex:
  \begin{itemize}
  \item If $u \in W$, then it corresponds to a color class $H$. The color
    separation property for $u$ means exactly that the removal of $H$ in $G$
    disconnects its vertex partition, which is what we wanted.
  \item Else, $u \in V$, and it can be seen that $u$ is also a color-separating
    vertex of~$G$. From this point on, we can forget the construction $G'$.
    Choose a color $c$ used by some edges incident to $u$, and let $H$ be the
    color class of $c$ in $G$ and $U$ be the vertices of $H$ except $u$.
    \begin{itemize}
    \item If $H$ is a star with center $c$ (which includes the case of a single
      edge), then \emph{by color separation}, removing the edges of $H$ indeed
      disconnects $u$ from $U$.
    \item Else, let $H'$ be the color class of $c$ in $G[V \setminus \{u\}]$;
      $H'$ has at least two vertices and is connected. Let $F$ be the connected
      component of $H'$ in $G[V \setminus \{u\}]$. $F$ is smaller than $G$, is
      also rainbow acyclic, and has at least one edge, so by the induction
      hypothesis, it has a color class with color $c'$ whose removal disconnects
      its vertex partition in $F$.
      \begin{itemize}
      \item If $c' \neq c$, note that the color component of $c'$ in $G$ is
        connected -- it is a complete multipartite graph -- and has no edge
        incident to $u$ -- again by color separation -- so it is actually
        included in $F$. The removal of this color component thus disconnects
        its vertex partition.
      \item If $c' = c$, then removing $H$ from $G$ disconnects both $u$ from
        $U$ and the vertices of $U$ which are not in the same part of the
        partition from each other; thus, we have what we wanted.
      \end{itemize}
    \end{itemize}
  \end{itemize}
\end{proof}

\section{Arc-colored directed graphs}

\subsection{Properly colored directed trails and closed trails}

We now consider \emph{directed} graphs (a.k.a.\ \emph{digraphs}) and trails
therein. It is important to note that given two vertices $u$ and $v$ in a
digraph, if the arcs $(u,v)$ and $(v,u)$ both exist, they are considered to be
different arcs, and therefore \emph{they can both appear in the same trail}.
(This follows the definitions used in~\cite{gourves_complexity_2013}.) In
particular, this means that a directed trail in a symmetric digraph -- symmetry
means that the arc $(u,v)$ is present if and only if $(v,u)$ is -- does not
necessarily induce a trail in the corresponding undirected graph.

This should help understand how the following property can hold while its
counterpart in edge-colored undirected graphs is completely wrong (the
quintessential counterexample being alternating walks for perfect matchings that
go through blossoms).

\begin{proposition}
  Let $G$ be an arc-colored digraph and $s,t$ be two vertices of $G$. From any
  properly colored directed \emph{walk} from $s$ to $t$ one can extract a subset
  of arcs which forms a properly colored directed \emph{trail} from $s$ to $t$.
\end{proposition}
\begin{proof}
  Consider a walk $s,e_1, \ldots, v_i,e_{i+1},v_{i+1}, \ldots,
  v_j,e_{j+1},v_{j+1},\ldots e_n, t$ where the $e_k$ are arcs, such that
  $e_{i+1} = e_{j+1}$. Since these arcs are directed, their sources are equal
  and their targets are equal: $v_i = v_j$, $v_{i+1} = v_{j+1}$. The following
  is therefore a legal walk: $s,e_1, \ldots, v_i,e_{i+1},v_{j+1},\ldots e_n, t$.
  If the initial walk was properly colored, then so is the new one ($e_{i+1} =
  e_{j+1}$ so in particular they have the same color). By iterating this process
  until we reach a minimal subwalk, we obtain a properly colored trail.
\end{proof}

This means that finding a PC directed trail is a particularly simple algorithmic
problem:
\begin{theorem}
  \label{thm:lolbfs}
  There is a linear time algorithm for finding directed properly colored trails.
\end{theorem}
\begin{proof}
  This can be done using a breadth-first search. Indeed, a BFS can be used to
  find PC walks of minimum length since the notion of PC walk, unlike that of PC
  path or trail, is \enquote{history-free} (paths and trails involve a global
  non-repetition constraint, while the constraints on PC walks are purely
  local). By the previous proposition, since this PC walk is minimal, it is a
  trail.
\end{proof}

\begin{remark}
  \label{rem:reload}
  Let us compare this with the polynomial time algorithm given in
  \cite[Theorem~1]{gourves_complexity_2013} (whose statement does not give a
  precise exponent). That algorithm ultimately proceeds by the following
  sequence of reductions:
  \[ \text{PC directed trail} \longrightarrow \text{minimum reload+weight
      directed trail} \longrightarrow \text{shortest weighted path/trail} \]
  The second reduction is treated in \cite[Proposition~1]{amaldi_minimum_2011}
  and involves a potentially quadratic blowup of the instance size. So the
  bounds that we can infer for the algorithm of
  \cite[Theorem~1]{gourves_complexity_2013} are quadratic -- worse than our
  linear time result.

  Note by the way that in the end, this algorithm relies on solving the shortest
  path problem via e.g.\ the classical Dijkstra algorithm, which is much simpler
  than the refinement of Edmonds's blossom algorithm used for most of our other
  results. This supports our idea that the problem is simple enough to be
  tackled by a mere breadth-first search.
\end{remark}

\subsection{Alternating circuits and directed paths}

Concerning properly colored directed paths, a NP-completeness result is known
even with only 2 colors and an acyclicity condition:
\begin{theorem}[{\cite[Theorem~5]{gourves_complexity_2013}}]
  Deciding whether a 2-arc-colored digraph contains a properly colored path
  between two given vertices is $\mathsf{NP}$-complete, even when the input is
  restricted to digraphs with no properly colored circuit.
\end{theorem}
We can deduce NP-hardness for PC circuits (directed cycles without vertex
repetitions).
\begin{corollary}
  \label{cor:pc-circuit}
  Finding a properly colored circuit in a 2-arc-colored digraph is
  $\mathsf{NP}$-complete.
\end{corollary}
\begin{proof}[Proof sketch]
  We prove NP-hardness by reduction from the previous problem (while membership
  in NP is trivial). To any 2-arc-colored digraph with designated source $s$ and
  designated target $t$, glue an acyclic gadget to $s$ and $t$ to add a properly
  arc-colored path from $t$ to $s$ with any starting color and any ending color.
  If the original digraph had no properly colored circuit, then the new one
  admits a properly colored circuit if and only if there was a properly colored
  path from $s$ to $t$ in the original, whose concatenation with a new path from
  $t$ to $s$ results in a circuit.
\end{proof}

Our goal here will be to show that a special case of this problem -- the
alternating circuit problem -- is already NP-complete. To define this restricted
problem, we must first explain what our notion of perfect matching is in the
setting of digraphs.

\begin{definition}
  A \emph{perfect matching} $M$ of a directed graph is a subset of arcs such
  that:
  \begin{itemize}
  \item any vertex $u \in V$ has exactly one outgoing arc in $M$ and exactly one
    incoming arc in $M$ (i.e.\ there is exactly one pair $(v,w) \in V^2$ such
    that $(u,v) \in M$ and $(w,u) \in M$);
  \item for all $u,v \in V$, $(u,v) \in M \iff (v,u) \in M$ -- morally, $M$
    consists of undirected edges.
  \end{itemize}
  \emph{Alternating paths} and \emph{alternating circuits} are defined as
  expected.
\end{definition}

\begin{theorem}
  \label{thm:alt-circuit-np}
  Detecting alternating circuits for perfect matchings in digraphs is
  NP-complete.
\end{theorem}
\begin{proof}
  It suffices to adapt in the obvious way Edmonds's reduction from
  2-edge-colored graphs to undirected graphs equipped with perfect matchings --
  recall that in \cref{sec:locally}, we presented and generalized this
  reduction. See \cref{fig:directed-reduction} for an example.

  As previously mentioned, we propose in~\cite{nguyen_proof_2020} an alternative
  proof which, instead of going through \cref{cor:pc-circuit}, proceeds by
  direct reduction from CNF-SAT.
\end{proof}

\begin{figure}
  \centering
    \begin{tikzpicture}
      \node[vertex] (w) at (0,2) {};
      \node[vertex] (x) at (4,2) {};
      \node[vertex] (y) at (0,0) {};
      \node[vertex] (z) at (4,0) {};
      \node[vertex] (o) at (2,1) {};
      \draw[lavenderindigo,thick,->] (x) -- (z);
      \draw[amber,thick,->] (x) -- (o);
      \draw[amber,thick,<-] (z) -- (o);
      \draw[lavenderindigo,thick,->] (w) -- (o);
      \draw[lavenderindigo,thick,<-] (y) -- (o);
    \end{tikzpicture}\quad\quad\quad\quad\begin{tikzpicture}
      \node[vertex,lavenderindigo] (w) at (1,2.5) {};
      \node[vertex,lavenderindigo] (x) at (5,2.5) {};
      \node[vertex,lavenderindigo] (y) at (1,1.5) {};
      \node[vertex,lavenderindigo] (z) at (5,1.5) {};
      \node[vertex,lavenderindigo] (o) at (3,2) {};

      \node[vertex,amber] (w') at (0,1) {};
      \node[vertex,amber] (x') at (4,1) {};
      \node[vertex,amber] (y') at (0,0) {};
      \node[vertex,amber] (z') at (4,0) {};
      \node[vertex,amber] (o') at (2,0.5) {};

      \draw[matching edge] (x) -- (x');
      \draw[matching edge] (y) -- (y');
      \draw[matching edge] (z) -- (z');
      \draw[matching edge] (w) -- (w');
      \draw[matching edge] (o) -- (o');

      \draw[non matching edge, thick, ->] (x) -- (z);
      \draw[non matching edge, thick, ->] (x') -- (o');
      \draw[non matching edge, thick, <-] (z') -- (o');
      \draw[non matching edge, thick, ->] (w) -- (o);
      \draw[non matching edge, thick, <-] (y) -- (o);
    \end{tikzpicture}
    \caption{A 2-arc-colored digraph and its translation into a digraph with a
      perfect matching.}
    \label{fig:directed-reduction}
\end{figure}
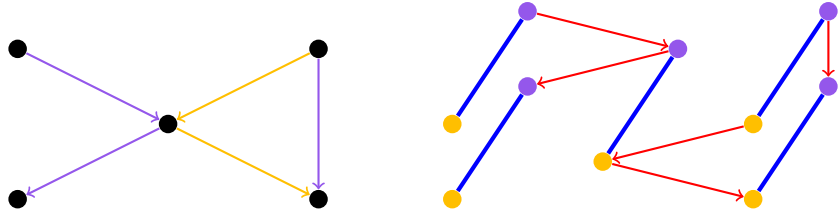

As we said in the introduction, while this theorem is not difficult and may
seem anecdotical, it has a potentially significant application in
logic~\cite{nguyen_proof_2020}. For this application, it is the hardness result
for circuits that is important. That said, for the sake of completeness, one can
also establish NP-completeness for alternating directed paths.

\begin{corollary}
  Finding an alternating directed path between two given vertices is
  NP-complete.
\end{corollary}
\begin{proof}
  Intuitively speaking, it is obvious that an oracle for alternating directed
  paths can be used to find alternating circuits. But to satisfy the definition
  of NP-hardness, we must exhibit a proper many-one reduction on instances.
  TODO.
\end{proof}

\bibliographystyle{alpha}
\bibliography{/home/tito/export.bib}

\end{document}